\documentclass[submission,copyright,creativecommons]{eptcs}
\providecommand{\event}{GANDALF 2022} 
\usepackage{breakurl}             
\usepackage{underscore}           

\usepackage{tikz,tikz-qtree}
\usetikzlibrary{arrows,positioning,calc}

\usepackage{wrapfig}
\usepackage{hhline}
\usepackage{multirow}
\usepackage{paralist}

\usepackage{multicol}
\usepackage{mathtools}
\usepackage{xspace}

\usepackage{amsmath,amsfonts,amssymb,amsthm}

\newcommand{\details}[1]{{}}
\newcommand{\tpl}[1]{(#1)}
\newcommand{\len}[1]{{\mathsf{len}_{#1}}}
\newcommand{\DefinedAs}{\ensuremath{\,\stackrel{\text{\textup{def}}}{=}\,}}

\newcommand{\Nat}{{\mathbb{N}}}
\newcommand{\INT}{{\mathbb{Z}}}
\newcommand{\NatP}{{\mathbb{N}_+}}

\newcommand{\Logic}{{\mathfrak{F}}}
\newcommand{\Prop}{\textit{AP}}
\newcommand{\Lang}{{\mathcal{L}}}
\newcommand{\LO}{\mathbb{U}}

\newcommand{\LOSupPt}{Pt}
\newcommand{\PU}{P_U}
\newcommand{\PL}{P_L}

\newcommand{\Intvs}{\mathbb{I}}
\newcommand{\IS}{\mathcal{S}}

 \newcommand{\Au}{\ensuremath{\mathcal{A}}}
\newcommand{\Ku}{\ensuremath{\mathcal{K}}}
\newcommand{\Lab}{\mathit{Lab}}
\newcommand{\Rel}{\mathit{rel}}
\newcommand{\Alt}{\mathit{alt}}

\newcommand{\HS}{\text{\sffamily HS}}
\newcommand{\HL}{\text{\sffamily HL}}
\newcommand{\PHS}{\text{\sffamily PHS}}
\newcommand{\PHSFrag}{\text{\sffamily P}(\mathsf{X_1\ldots X_n})}
\newcommand\PHSF[1]{\text{\sffamily P}(#1)}
\newcommand{\PromptHS}{\text{\sffamily PromptHS}}
\newcommand{\PromptFrag}{\text{\sffamily Prompt}(\mathsf{X_1\ldots X_n})}
\newcommand\PromptF[1]{\text{\sffamily Prompt}(#1)}
\newcommand{\MTL}{\text{\sffamily MTL}}
  
\newcommand{\CTL}{\text{\sffamily CTL}}
\newcommand{\CTLStar}{\text{\sffamily CTL$^{*}$}}
\newcommand{\LTL}{\text{\sffamily LTL}}
\newcommand{\PLTL}{\text{\sffamily PLTL}}
\newcommand{\NFA}{\text{\sffamily NFA}}
\newcommand{\FO}{\text{\sffamily FO}}

\newcommand{\RelA}{\ensuremath{\mathcal{R}_A}}
\newcommand{\RelL}{\ensuremath{\mathcal{R}_L}}
\newcommand{\RelB}{\ensuremath{\mathcal{R}_B}}
\newcommand{\RelO}{\ensuremath{\mathcal{R}_O}}
\newcommand{\RelE}{\ensuremath{\mathcal{R}_E}}
\newcommand{\RelD}{\ensuremath{\mathcal{R}_D}}
\newcommand{\RelX}{\ensuremath{\mathcal{R}_X}}
\newcommand{\RelXt}{\ensuremath{\mathcal{R}_{\overline{X}}}}
\newcommand{\RelBt}{\ensuremath{\mathcal{R}_{\overline{B}}}}
\newcommand{\RelEt}{\ensuremath{\mathcal{R}_{\overline{E}}}}
 \newcommand{\RelXP}[1]{\ensuremath{\mathcal{R}_{X_{#1}}}}

\DeclareMathOperator{\hsX}{\langle X\rangle}
\DeclareMathOperator{\hsA}{\langle A\rangle}
\DeclareMathOperator{\hsL}{\langle L\rangle}
\DeclareMathOperator{\hsB}{\langle B\rangle}
\DeclareMathOperator{\hsE}{\langle E\rangle}
\DeclareMathOperator{\hsD}{\langle D\rangle}
\DeclareMathOperator{\hsO}{\langle O\rangle}

\DeclareMathOperator{\hsAt}{\langle \overline{A}\rangle}
\DeclareMathOperator{\hsLt}{\langle \overline{L}\rangle}
\DeclareMathOperator{\hsBt}{\langle \overline{B}\rangle}
\DeclareMathOperator{\hsBtW}{\langle \overline{B}_w\rangle}
\DeclareMathOperator{\hsEt}{\langle \overline{E}\rangle}
\DeclareMathOperator{\hsEtW}{\langle \overline{E}_w\rangle}
\DeclareMathOperator{\hsDt}{\langle \overline{D}\rangle}
\DeclareMathOperator{\hsOt}{\langle \overline{O}\rangle}

\DeclareMathOperator{\hsUX}{[X]}

\DeclareMathOperator{\hsUA}{[A]}

\DeclareMathOperator{\hsUB}{[B]}
\DeclareMathOperator{\hsUE}{[E]}

\DeclareMathOperator{\hsUBt}{[\overline{B}]}
\DeclareMathOperator{\hsUBtW}{[\overline{B}_w]}
\DeclareMathOperator{\hsUEt}{[\overline{E}]}
\DeclareMathOperator{\hsUEtW}{[\overline{E}_w]}

\newcommand{\AB}{\mathsf{AB}}

\newcommand{\AAbar}{\mathsf{A\overline{A}}}

\newcommand{\ABBbarBbarW}{\mathsf{AB\overline{B}\overline{B}_w}}

\newcommand{\BBbarBbarWEEbarEbarW}{\mathsf{B\overline{B}\overline{B}_w E\overline{E}\overline{E}_w}}
\newcommand{\BBbarEEbar}{\mathsf{B\overline{B} E\overline{E}}}

\newcommand{\B}{\mathsf{B}}

\newcommand{\D}{\mathsf{D}}
\newcommand{\BE}{\mathsf{BE}}

\newcommand{\ABBbar}{\mathsf{AB\overline{B}}}

\newcommand{\until}{\textsf{U}}
\newcommand{\Next}{\textsf{X}}
\newcommand{\PNext}{\textsf{Y}}

\newcommand{\Always}{\textsf{G}}
\newcommand{\Eventually}{\textsf{F}}
\newcommand{\PEventually}{\textsf{P}}

\newcommand{\DBinder}{\text{$\downarrow$$x$}}

\def\NLOGSPACE{{\sc NLogspace}}

\def\PSPACE{{\sc Pspace}}
\def\EXPSPACE{{\sc Expspace}}
\def\NEXPSPACE{{\sc NExpspace}}

\newcommand{\PTIME}{\mathbf{P}}
\newcommand{\NP}{\mathbf{NP}}
\newcommand{\coNP}{\mathbf{co-NP}}

 \newtheorem{proposition}{Proposition}

      \newtheorem{theorem}{Theorem}
    \newtheorem{lemma}{Lemma}
    \newtheorem{corollary}{Corollary}

\title{Parametric Interval Temporal Logic over Infinite Words}
\author{Laura Bozzelli
\institute{University of Napoli ``Federico II'', Napoli, Italy}
\email{laura.bozzelli@unina.it}
\and
Adriano Peron
\institute{University of Napoli ``Federico II'', Napoli, Italy}
\email{adrperon@unina.it}
}
\def\titlerunning{Parametric Interval Temporal Logic over Infinite Words}
\def\authorrunning{L. Bozzelli \& A. Peron}

\begin{document}

\maketitle

\begin{abstract}
Model checking  for Halpern and Shoham's interval temporal logic \HS\ has been recently investigated in a systematic way, and it is known to be decidable under three distinct semantics.  
Here, we focus on the \emph{trace-based semantics}, where the  infinite execution paths (traces) of the given (finite) Kripke structure are the main semantic entities. In this setting, each finite infix of a trace is interpreted as an interval, and a proposition holds over an interval if and only if it holds over each component state (\emph{homogeneity assumption}).
In this paper, we introduce a quantitative extension of $\HS$ over traces, called \emph{parametric} $\HS$ (\PHS). The novel logic
allows to express parametric timing constraints on the duration (length) of the intervals. We show that 
checking the existence of a parameter valuation for which a Kripke structure satisfies a $\PHS$ formula (model checking), or a $\PHS$ formula admits a trace as a model under the homogeneity assumption (satisfiability) is decidable. 
Moreover, we identify a fragment of $\PHS$ which subsumes parametric $\LTL$ and for which model checking and satisfiability 
are shown to be \EXPSPACE-complete. 
\end{abstract}

\section{Introduction}

\textbf{Interval temporal logic $\HS$.} \emph{Point-based} Temporal Logics (PTLs), such as  the linear-time temporal logic $\LTL$~\cite{pnueli1977temporal} and the branching-time temporal logics $\CTL$ and $\CTLStar$~\cite{emerson1986sometimes}
provide a standard framework for the specification of the dynamic behavior
of reactive systems that makes it possible to describe how a system evolves state-by-state
(``point-wise" view). PTLs have been successfully employed in model checking (MC)~\cite{Clarke81ctl,Queille81verification} for the
automatic verification of complex finite-state systems modeled as finite propositional
Kripke structures.
\emph{Interval Temporal Logics} (ITLs) provide an alternative setting for reasoning about time~\cite{HS91,DigitalCircuitsThesis,Ven90}.
They assume intervals, instead of points, as their primitive temporal entities
allowing one to specify temporal properties that involve, e.g., actions with duration, accomplishments, and temporal aggregations, which are inherently ``interval-based'', and thus cannot be naturally expressed by PTLs. ITLs find  applications in a variety of computer
science fields, including  artificial intelligence (reasoning about action and
change, qualitative reasoning, planning, and natural language processing),
theoretical computer science (specification and verification of programs),
and temporal and spatio-temporal databases (see, e.g.,~\cite{DigitalCircuitsThesis,LM13,Pratt-Hartmann05}).

The most prominent example of ITLs is \emph{Halpern and Shoham's modal logic of time intervals} (\HS)~\cite{HS91}
 which features one modality for each of the 13 possible ordering
	relations between pairs of intervals (the so-called Allen's
	relations~\cite{All83}), apart from equality.
The satisfiability problem for \HS\ turns out to be highly undecidable for all interesting (classes of) linear orders~\cite{HS91}. The same happens with most of its fragments~\cite{BresolinMGMS14,Lod00,MM14} with some meaningful exceptions 
like the logic of temporal neighbourhood $\AAbar$, over all relevant (classes of) linear orders~\cite{BresolinMSS11}, and the logic of sub-intervals $\D$, over the class of dense linear orders~\cite{
MontanariPS15}.

Model checking of (finite) Kripke structures against \HS\ has been investigated only  recently~\cite{LM13,LM14,LM16,MolinariMMPP16, MolinariMPS16, BozzelliMMPS18,BozzelliMMPS19,BozzelliMMP20,BozzelliMMPS22}.
The idea is to  interpret each finite path of a Kripke structure as an interval, whose labelling is defined on the basis of the labelling of the component states, that is, a proposition letter holds over an interval if and only if it holds over each component state (\emph{homogeneity assumption}~\cite{Roe80}).
Most of the results have been obtained by adopting the so-called \emph{state-based semantics}~\cite{MolinariMMPP16}:
intervals/paths are ``forgetful" of the history leading to their starting state, and time branches both in the future and in the past. In this setting, MC of full $\HS$ is decidable: the problem  is at least
\EXPSPACE-hard~\cite{BMMPS16}, while the only known upper bound is non-elementary~\cite{MolinariMMPP16}.
The known complexity bounds  for full $\HS$ coincide with those for the linear-time fragment $\BE$ of $\HS$ which features modalities $\hsB$  and $\hsE$ for prefixes and suffixes.
These complexity bounds easily transfer  to finite satisfiability, that is, satisfiability over finite linear orders, of $\BE$ under the homogeneity assumption. Whether or not these problems can be solved elementarily is a difficult open question.
On the other hand, in the state-based setting, the exact complexity of MC for many meaningful (linear-time or branching-time) syntactic
 fragments of $\HS$, which ranges from $\coNP$ to $\PTIME^{\NP}$, \PSPACE, and beyond, has been determined in
a series of papers~\cite{MolinariMPS16,BozzelliMMPS18,BozzelliMMPS19b,BozzelliMP21,BozzelliMPS21,BozzelliMMPS22}.

The
expressiveness of $\HS$ with the state-based semantics  has been studied in~\cite{BozzelliMMPS19}, together  with other two decidable variants:
	the \emph{computation-tree-based semantics} variant and the \emph{traces-based} one. For the first variant, past is linear: each
interval may have several possible futures, but only a unique past. Moreover, past is finite and cumulative, and is never forgotten.
The trace-based approach instead relies on a linear-time setting, where the infinite paths (traces) of the given Kripke structure are the main semantic entities.
It is known that the computation-tree-based variant of $\HS$ is expressively equivalent to finitary  $\CTLStar$ (the variant of $\CTLStar$ with quantification over finite paths), while the trace-based variant is equivalent to $\LTL$. The state-based variant is more expressive than the computation-tree-based variant and expressively incomparable with both $\LTL$ and $\CTLStar$. To the best of our knowledge, complexity issues about MC and the satisfiability problem of
$\HS$ and its syntactic fragments under the trace-based semantics have not been investigated so far.\vspace{0.1cm}

\textbf{Parametric extensions of point-based temporal logics.} Traditional PTLs such as standard $\LTL$~\cite{pnueli1977temporal} allow only to express \emph{qualitative} requirements on the temporal ordering of events. For example, in expressing a typical request-response temporal requirement, it is not possible to specify a bound on the amount of time for which a request is granted. A simple way to overcome this drawback is to consider quantitative extensions of PTLs where temporal modalities are equipped with timing constraints for allowing the specification of \emph{constant} bounds on the delays among events. A well-known representative of such logics is Metric Temporal Logic (\MTL)~\cite{Koymans90}.    However this approach is not practical in the first stages of a design, when not much is known about the system under development, and is useful for designers to use parameters instead of specific constants.  Parametric extensions of traditional PTLs, where time bounds can be expressed by means of parameters, have been investigated in many papers. Relevant examples include parametric $\LTL$~\cite{AlurETP01}, Prompt $\LTL$~\cite{KupfermanPV09}, and parametric $\MTL$~\cite{GiampaoloTN10}.
\vspace{0.1cm}

\textbf{Our contribution.} In this paper we introduce a parametric extension of the interval temporal logic $\HS$ under the trace-based semantics, called
\emph{parametric} $\HS$ ($\PHS$). The extension is obtained by means of inequality constraints on the temporal modalities of $\HS$
which allow to specify parametric lower/upper bounds on the duration (length) of the interval selected by the temporal modality.
Similarly to parametric $\LTL$~\cite{AlurETP01}, we impose that a parameter can be exclusively used either as upper bound or as lower bound in the timing constraints. We address the decision problems of checking the existence of a parameter valuation such that 
(1) a given $\PHS$ formula is satisfiable, and (2) a given Kripke structure satisfies a given $\PHS$ formula (MC). By adapting the
alternating color technique for Prompt $\LTL$~\cite{KupfermanPV09} and by exploiting known results on \emph{linear-time hybrid logic $\HL$}~\cite{FRS03,SW07,BozzelliL10}, we show that the considered problems are decidable. Additionally, we consider the syntactic fragment $\PHSF{\ABBbar}$ of $\PHS$ which allows only temporal modalities for the Allen's relations \emph{meets} $\RelA$, \emph{started-by} $\RelB$ and its inverse $\RelBt$. We show that $\PHSF{\ABBbar}$ subsumes parametric $\LTL$, and its flat fragment
$\ABBbar$ is exponentially more succinct than $\LTL$ + past. Moreover, we establish that satisfiability and MC of $\PHSF{\ABBbar}$
are \EXPSPACE-complete, and we provide tight bounds on optimal parameter values for both  problems.    
\section{Preliminaries}\label{sec:preliminary}

We fix the following notation. Let $\INT$ be the set of integers, $\Nat$ the set of natural numbers, and $\NatP\DefinedAs \Nat\setminus \{0\}$.
Let $\Sigma$ be an alphabet and $w$ be a non-empty finite or infinite word over $\Sigma$. We denote by $|w|$ the length of $w$ ($|w|=\infty$ if $w$ is infinite). For all  $i,j\in\Nat $, with $i\leq j<|w|$, $w(i)$ is the
$(i+1)$-th letter of $w$, while $w[i,j]$ is the infix of $w$ given by $w(i)\cdots w(j)$.

\noindent We fix a finite set $\Prop$ of atomic propositions. A \emph{trace} is an infinite word over $2^{\Prop}$.
For a logic $\Logic$ interpreted  over traces and a formula $\varphi\in\Logic$,  $\Lang(\varphi)$ denotes the
set of traces satisfying $\varphi$. The \emph{satisfiability problem for} $\Logic$ is checking for a given formula $\varphi\in\Logic$,
whether $\Lang(\varphi)\neq \emptyset$.
  \vspace{0.2cm}

\noindent\textbf{Kripke Structures.} In the context of model-checking, finite state systems are usually modelled as finite Kripke structures over a finite set $\Prop$ of atomic
propositions which represent predicates over the states of the system.
A \emph{(finite) Kripke structure} over   $\Prop$   is a tuple  $\Ku=\tpl{\Prop,S, E,\Lab,s_0}$, where  $S$ is a finite set of states,
$E\subseteq S\times S$ is a left-total transition relation, $\Lab:S\mapsto 2^{\Prop}$ is a labelling function assigning to each state $s$ the set of propositions that hold over it, and $s_0\in S$ is the initial state. 
 An infinite  path $\pi$ of $\Ku$ is an infinite word over $S$ such that $\pi(0)=s_0$  and $(\pi(i),\pi(i+1))\in E$ for all $i\geq 0$. A finite path of $\Ku$ is a non-empty infix of some infinite path of $\Ku$.
An infinite  path $\pi$ induces the  trace given by
$\Lab(\pi(0))\Lab(\pi(1))\ldots$. We denote  by $\Lang(\Ku)$ the set of traces associated with the infinite paths of $\Ku$.
Given a logic $\Logic$ interpreted  over traces, the \emph{(linear-time) model checking problem against $\Logic$} is checking for a given  Kripke structure $\Ku$ and a formula $\varphi\in \Logic$, whether $\Lang(\Ku)\subseteq \Lang(\varphi)$.   \vspace{0.2cm}

\noindent\textbf{B\"{u}chi nondeterministic automata.} A \emph{B\"{u}chi nondeterministic finite automaton}  over infinite words (B\"{u}chi $\NFA$ for short) is a tuple $\Au=\tpl{\Sigma,Q,q_0,\delta,F}$, where $\Sigma$ is a finite input alphabet, $Q$ is a finite set of states,
$q_0\in Q$ is the initial state, $\delta:Q\times \Sigma \mapsto 2^{Q}$ is the transition relation,  and $F\subseteq Q$ is a set of accepting states. Given an infinite word $w$ over $\Sigma$, a run $\pi$
of $\Au$ over $w$ is a an infinite sequence $\pi$ of states such that $\pi(0)=q_0$ and $\pi(i+1)\in\delta(\pi(i),w(i))$ for all $i\geq 0$. The run is accepting if for infinitely many $i\geq 0$, $\pi(i)\in F$.  The language $\Lang(\Au)$ accepted by $\Au$ is the set of infinite words $w$
over $\Sigma$ such that there is an accepting run of $\Au$ over $w$.

\subsection{Allen's relations and Interval Temporal Logic $\HS$}

An interval algebra to reason about intervals and their relative orders was proposed by Allen in~\cite{All83}, while a systematic logical study of interval representation and reasoning was done a few years later by Halpern and Shoham, who introduced the interval temporal logic $\HS$ featuring one modality for each Allen
relation, but equality~\cite{HS91}.

Let $\LO = \tpl{\LOSupPt,<}$ be a  linear order over the nonempty set $\LOSupPt\neq \emptyset$, and $\leq$ be the reflexive closure of $<$. Given two elements $x,y\in \LOSupPt$ such that $x\leq y$, we denote by $[x,y]$
the (non-empty closed) \emph{interval} over $\LOSupPt$ given by the set of elements $z\in \LOSupPt$ such that $x\leq z$ and $z\leq y$.
We denote the set of all intervals over
$\LO$ by $\Intvs(\LO)$.
 We now recall the Allen's relations over intervals of the linear order $\LO = \tpl{\LOSupPt,<}$:
\begin{compactenum}
  \item the \emph{meet} relation $\RelA$, defined
    by $[x, y]\,\RelA\,  [v, z]$ if $y=v$ (i.e., the start-point of the second interval coincides
    with the end-point of the first interval);
      \item the \emph{before} relation $\RelL$, defined
    by $[x, y]\,\RelL \, [v, z]$ if $y<v$ (i.e., the start-point of the second interval strictly follows
   the end-point of the first interval);
    \item the \emph{started-by} relation $\RelB$, defined
    by $[x, y]\,\RelB\,  [v, z]$ if $x =v$ and $z<y$ (i.e., the  second interval is a proper prefix
   of the first interval);
   \item the \emph{finished-by} relation $\RelE$, defined
    by $[x, y]\,\RelE\,  [v, z]$ if $y =z$ and $x<v$ (i.e., the  second interval is a proper suffix
   of the first interval);
    \item the \emph{contains} relation $\RelD$, defined
    by $[x, y]\,\RelD\,  [v, z]$ if $x< v$ and $z<y$ (i.e., the  second interval is contained in the internal of the first interval);
       \item the \emph{overlaps} relation $\RelO$, defined
    by $[x, y]\,\RelO\,  [v, z]$ if $x<v<y<z$ (i.e., the second interval overlaps at the right the first interval);
    \item for each $X\in \{A,L,B,E,D,O\}$ the relation $\RelXt$, defined as the inverse of $\RelX$, i.e.
      $[x, y]\,\RelXt\,  [v, z]$ if $[v, z]\RelX   [x, y]$.
\end{compactenum}\vspace{0.1cm}

Table~\ref{allen} gives a graphical representation of the Allen's relations $\RelA$, $\RelL$, $\RelB$, $\RelE$,
$\RelD$, and $\RelO$ together with the corresponding $\HS$ (existential) modalities.
\begin{table}[tb]
\centering
\caption{Allen's relations and corresponding $\HS$ modalities.}\label{allen}
\vspace*{0.1cm}
{
\begin{tabular}{cclc}
\hline
\rule[-1ex]{0pt}{3.5ex} Allen relation & $\HS$ & Definition w.r.t. interval structures &  Example\\
\hline

&   &   & \multirow{7}{*}{\begin{tikzpicture}[scale=0.785]
\draw[draw=none,use as bounding box](-0.3,0.2) rectangle (3.3,-3.1);
\coordinate [label=left:\textcolor{red}{$x$}] (A0) at (0,0);
\coordinate [label=right:\textcolor{red}{$y$}] (B0) at (1.5,0);
\draw[red] (A0) -- (B0);
\fill [red] (A0) circle (2pt);
\fill [red] (B0) circle (2pt);

\coordinate [label=left:$v$] (A) at (1.5,-0.5);
\coordinate [label=right:$z$] (B) at (2.5,-0.5);
\draw[black] (A) -- (B);
\fill [black] (A) circle (2pt);
\fill [black] (B) circle (2pt);

\coordinate [label=left:$v$] (A) at (2,-1);
\coordinate [label=right:$z$] (B) at (3,-1);
\draw[black] (A) -- (B);
\fill [black] (A) circle (2pt);
\fill [black] (B) circle (2pt);

\coordinate [label=left:$v$] (A) at (0,-1.5);
\coordinate [label=right:$z$] (B) at (1,-1.5);
\draw[black] (A) -- (B);
\fill [black] (A) circle (2pt);
\fill [black] (B) circle (2pt);

\coordinate [label=left:$v$] (A) at (0.5,-2);
\coordinate [label=right:$z$] (B) at (1.5,-2);
\draw[black] (A) -- (B);
\fill [black] (A) circle (2pt);
\fill [black] (B) circle (2pt);

\coordinate [label=left:$v$] (A) at (0.5,-2.5);
\coordinate [label=right:$z$] (B) at (1,-2.5);
\draw[black] (A) -- (B);
\fill [black] (A) circle (2pt);
\fill [black] (B) circle (2pt);

\coordinate [label=left:$v$] (A) at (1.3,-3);
\coordinate [label=right:$z$] (B) at (2.3,-3);
\draw[black] (A) -- (B);
\fill [black] (A) circle (2pt);
\fill [black] (B) circle (2pt);

\coordinate (A1) at (0,-3);
\coordinate (B1) at (1.5,-3);
\draw[dotted] (A0) -- (A1);
\draw[dotted] (B0) -- (B1);
\end{tikzpicture}}\\

\textsc{meets} & $\hsA$ & $[x,y]\,\RelA\,[v,z]\iff y=v$ &\\

\textsc{before} & $\hsL$ & $[x,y]\,\RelL\,[v,z]\iff y<v$ &\\

\textsc{started-by} & $\hsB$ & $[x,y]\,\RelB\,[v,z]\iff x=v\wedge z<y$ &\\

\textsc{finished-by} & $\hsE$ & $[x,y]\,\RelE\,[v,z]\iff y=z\wedge x<v$ &\\

\textsc{contains} & $\hsD$ & $[x,y]\,\RelD\,[v,z]\iff x<v\wedge z<y$ &\\

\textsc{overlaps} & $\hsO$ & $[x,y]\,\RelO\,[v,z]\iff x<v<y<z$ &\\

\hline
\end{tabular}}
\end{table}\vspace{0.2cm}

\noindent \textbf{Syntax and semantics of $\HS$.}  $\HS$ formulas  $\varphi$ over $\Prop$
are defined as follows:
\[
   \varphi ::= \top \;\vert\;   p \;\vert\; \neg\varphi \;\vert\; \varphi \wedge \varphi \;\vert\; \hsX \varphi
\]
 where $p\in\Prop$ and $\hsX$ is the existential temporal modality  for the   (non-trivial)
Allen's relation $\RelX$, where $X\in\{A,L,B,E,D,O,\overline{A},\overline{L},\overline{B},\overline{E},\overline{D},\overline{O}\}$.
The size $|\varphi|$ of a formula $\varphi$ is the number of distinct subformulas of $\varphi$.
We  also exploit the standard logical connectives $\vee$ (disjunction) and $\rightarrow$ (implication)  as abbreviations,
and for any temporal  modality $\hsX$, the dual universal modality $\hsUX$  defined as: $\hsUX\psi\DefinedAs \neg\hsX\neg\psi$.
Moreover, we will also use the reflexive closure of the Allen's relation $\RelBt$ (resp., $\RelEt$) and the associated temporal modalities $\hsBtW$ and $\hsUBtW$ (resp., $\hsEtW$ and $\hsUEtW$) where
$\hsBtW\varphi$ corresponds to $\varphi\vee \hsBt\varphi$ and  $\hsEtW\varphi$ corresponds to $\varphi\vee \hsEt\varphi$.
Given any subset of Allen's relations $\{\RelXP{1},..,\RelXP{n}\}$, we denote by $\mathsf{X_1 \cdots X_n}$ the \HS\ fragment featuring temporal modalities for $\RelXP{1},..,\RelXP{n}$ only.

The   logic $\HS$ is interpreted on \emph{interval structures} $\IS=\tpl{\Prop,\LO,\Lab}$, which are linear orders $\LO$ equipped with a labelling function $\Lab: \Intvs(\LO) \to 2^{\Prop}$ assigning to each interval the set of propositions that hold over it. Given an $\HS$ formula $\varphi$ and an interval $I  \in \Intvs(\LO)$, the satisfaction relation $I\models_\IS \varphi$, meaning that $\varphi$ holds at the interval $I$ of $\IS$, is inductively defined as follows (we omit the semantics of the Boolean connectives which is standard):
 \[
 \begin{array}{ll}
I \models_\IS  p  &  \Leftrightarrow  p\in \Lab(I); \\
I\models_\IS     \hsX \varphi  &  \Leftrightarrow \text{there is an interval $J\in \Intvs(\LO)$ such that $I\, \RelX \,J$ and }   J\models_\IS \varphi.
\end{array}
\]
It is worth noting that  we assume the \emph{non-strict semantics of $\HS$},
which admits intervals  consisting of a single point. Under such an assumption, all $\HS$-temporal modalities  can be expressed in terms of $\hsB, \hsE, \hsBt$, and $\hsEt$ (see~\cite{Ven90}).
As an example,  $\hsD\varphi$ can be expressed in terms of $\hsB$ and $\hsE$ as $ \hsB\hsE\varphi$, while  $\hsA\varphi$ can be expressed in terms of $\hsE$ and $\hsBt$ as
\[
 (\hsUE\,\neg\top \wedge (\varphi \vee \hsBt \varphi)) \vee \hsE (\hsUE\,\neg\top \wedge (\varphi \vee \hsBt \varphi)).
 \]

\noindent \textbf{Interpretation of $\HS$ over traces.}
In this paper, we focus on interval structures $\IS=\tpl{\Prop,(\Nat,<),\Lab}$ over the standard linear order on $\Nat$ ($\Nat$-interval structures for short) satisfying the \emph{homogeneity principle}: a proposition holds over an interval if
and only if it holds over all its subintervals. Formally, $\IS$ is \emph{homogeneous} if
 for every interval $[i,j]$ over $\Nat$ and every $p \in \Prop$,
it holds that $p\in \Lab([i,j])$ if and only if $p\in \Lab([h,h])$
for every $h\in [i,j]$.
Note that homogeneous $\Nat$-interval structures over $\Prop$ correspond to traces where, intuitively, each interval is mapped to an infix of the trace. Formally, each trace $w$  induces the homogeneous $\Nat$-interval structure $\IS(w)$ whose labeling function $\Lab_w$ is defined as follows:
for all $i,j\in\Nat$ with $i\leq j$ and $p\in\Prop$,  $p\in \Lab_w([i,j])$ if and only if $p\in w(h)$ for all $h\in [i,j]$. For the given finite set $\Prop$ of atomic propositions, this mapping from traces to homogeneous $\Nat$-interval structures is evidently a bijection.
For a trace $w$, an interval $I$ over $\Nat$, and an $\HS$ formula $\varphi$, we write $I\models_w \varphi$ to mean that $I\models_{\IS(w)} \varphi$. The trace $w$ satisfies $\varphi$, written $w\models \varphi$, if $[0,0]\models_w \varphi$.\vspace{0.2cm}

\noindent \textbf{Expressiveness completeness and succinctness of the fragment $\AB$ over traces.}
It is known that $\HS$ over traces has the same expressiveness as standard $\LTL$~\cite{BozzelliMMPS19}, where the latter is expressively complete
for standard first-order logic $\FO$ over traces~\cite{kamp1968tense}. In particular, the fragment $\AB$ of $\HS$ is sufficient for capturing full $\LTL$~\cite{BozzelliMMPS19}: given an $\LTL$ formula, one can construct in linear-time an equivalent
$\AB$ formula~\cite{BozzelliMMPS19}. Note that when interpreted on infinite words $w$, modality
 $\hsB$   allows to select proper non-empty prefixes   of the current infix subword of $w$, while modality
  $\hsA$   allows to select subwords whose first position coincides with the last position of the current interval.
 Here, we show that $\AB$ is exponentially more succinct than $\LTL$ + past. For each $k\geq 1$, we denote by $\len{k}$  the $\B$ formula capturing the intervals of length $k$: $\len{k}\DefinedAs (\underbrace{\hsB \ldots \hsB}_{\text{$k-1$ times}}\top)\wedge (\underbrace{\hsUB \ldots \hsUB}_{\text{$k$ times}}\neg\top)$.

For each $n\geq 1$, let $\Prop_n=\{p_0,\ldots,p_n\}$ and $\Lang_n$ be the $\omega$-language consisting of the infinite words
over $2^{\Prop_n}$ such that any two positions that agree on the truth value of propositions $p_1,\ldots,p_n$ also agree on the truth value of $p_0$. It is known that any B\"{u}chi $\NFA$ accepting $\Lang_n$ needs at least $2^{2^{n}}$ states~\cite{EtessamiVW02}. Thus, since any formula $\varphi$ of $\LTL$ + past can be translated into an equivalent B\"{u}chi $\NFA$ with a single exponential blow-up, it follows that any formula of $\LTL$ + past capturing $\Lang_n$ has size at least single exponential in $n$. On the other hand, the language $\Lang_n$ is captured by the following $\AB$ formula having size linear in $n$:
\[
\hsUA \hsUA \bigl(\displaystyle{(\bigwedge_{i\in [1,n]}}\theta(p_i)) \rightarrow \theta(p_0)  \bigr)\quad\quad
\theta(p)\DefinedAs \hsB(\len{1}\wedge p) \leftrightarrow \hsA(\len{1}\wedge p)
\]
Hence, we obtain the following result.

\begin{theorem} $\AB$ (over traces) is exponentially more succinct than $\LTL$ + past.
\end{theorem} 
\section{Parametric Interval Temporal Logic}\label{sec:ParametricHS}

 In this section, we introduce a parametric extension of the  interval temporal logic $\HS$ over traces, called \emph{parametric $\HS$} ($\PHS$ for short).
 The extension is obtained by means of inequality constraints on the temporal modalities of $\HS$ which allow  to compare
 the length of the interval selected by the temporal modality  with an integer parameter.
Like parametric $\LTL$~\cite{AlurETP01}, the  parameterized operators are monotone (either upward or downward) and a parameter is upward (resp., downward) if it is the subscript of some upward (resp., downward) modality. \vspace{0.2cm}

\noindent \textbf{Syntax and semantics of $\PHS$} Let $\PU$ be a finite set of \emph{upward} parameter variables $u$
and $\PL$ be a finite set of \emph{downward} parameter variables $\ell$ such that $\PU$ and $\PL$ are disjunct.  The syntax of $\PHS$ formulas $\varphi$ over $\Prop$  and the set $\PU\cup \PL$ of parameter variables is given in positive normal form as follows:
\[
\varphi ::= \top \;\vert\;   p \;\vert\; \neg p \;\vert\; \varphi \vee \varphi \;\vert\; \varphi \wedge \varphi \;\vert\; \hsX  \varphi  \;\vert\; \hsX_{\prec  u} \varphi \;\vert\; \hsX_{\succ   \ell} \varphi \;\vert\; \hsUX\varphi \;\vert\; \hsUX_{\prec  \ell} \varphi \;\vert\; \hsUX_{\succ   u} \varphi
\]
\noindent where $p\in\Prop$, $X\in\{A,L,B,E,D,O,\overline{A},\overline{L},\overline{B},\overline{B}_w,\overline{E}_w,\overline{D},\overline{O}\}$, $\prec\in\{<,\leq\}$, $\succ\in\{>,\geq\}$, $u\in \PU$, and $\ell\in \PL$.
We denote by $\PromptHS$ the fragment of $\PHS$ where the unique parameterized temporal modalities are of the form
 $\hsX_{\prec  u}$.
Moreover, given any subset of Allen's relations  $\{\RelXP{1},..,\RelXP{n}\}$, we denote by $\PHSFrag$ (resp., $\PromptFrag$) the $\PHS$ (resp., $\PromptHS$) fragment
featuring   temporal modalities for $\RelXP{1},..,\RelXP{n}$  only. We will focus on $\PHS$ and the fragment
$\PHSF{\ABBbarBbarW}$. 

For an interval $I=[i,j]$ over $\Nat$, we denote by $|I|$ the length of $I$, given by $j-i+1$. The semantics of a $\PHS$ formula $\varphi$ is inductively defined with respect to a trace $w$, an interval $I$ over $\Nat$, and a \emph{parameter valuation} $\alpha: \PU\cup \PL \mapsto \NatP$ assigning to each parameter variable a positive integer. We write $(I,\alpha) \models_w \varphi$ to  mean  that $\varphi$ holds at the interval $I$ of $w$ under the valuation $\alpha$. The interpretation of all temporal operators of $\HS$ and connectives is identical to their $\HS$ interpretations. The parameterized operators are interpreted as follows, where $\wp \in \PU\cup \PL$ and $\sim\in \{<,\leq, >,\geq\}$:
 \[
 \begin{array}{ll}
(I, \alpha) \models_w  \hsX_{\sim  \wp} \varphi  &  \Leftrightarrow  \text{there is some interval $J$ such that $I\, \RelX \,J$, $|J|\sim \alpha(\wp)$, and }   J, \alpha \models_w   \varphi;\\
(I, \alpha) \models_w  \hsUX_{\sim  \wp} \varphi  &  \Leftrightarrow  \text{for each interval $J$ such that $I\, \RelX \,J$ and $|J|\sim \alpha(\wp)$: }   J, \alpha \models_w   \varphi.
\end{array}
\]
We say that the trace $w$ is a model of formula $\varphi$ under the parameter valuation $\alpha$, written $(w,\alpha)\models \varphi$, if $([0,0],\alpha)\models_w \varphi$. For a $\PHS$ formula $\varphi$ and a  Kripke structure $\Ku$ over $\Prop$, we consider:
\begin{compactenum}
\item [(i)] the set $V(\Ku,\varphi)$ consisting of the parameter valuations $\alpha$ such that for each trace $w\in\Lang(\Ku)$ of $\Ku$, $(w,\alpha)\models \varphi$, and 
\item [(ii)] the set $S(\varphi)$ consisting of the valuations $\alpha$ such that
$(w,\alpha)\models \varphi$ for some trace $w$.
\end{compactenum}
 The \emph{(linear-time) model-checking problem against $\PHS$} is checking for a given  Kripke structure $\Ku$ and $\PHS$ formula $\varphi$ whether $V(\Ku,\varphi)\neq \emptyset$.  The \emph{satisfiability problem against $\PHS$} is checking for a given $\PHS$ formula $\varphi$ whether $S(\varphi)\neq \emptyset$.

Given two valuations $\alpha$ and $\beta$, we write $\alpha\leq \beta $ to mean that $\alpha(\wp)\leq \beta(\wp)$ for all
$\wp\in \PL\cup \PU$. A parameterized operator $\Theta$ is \emph{upward-monotone} (resp., \emph{downward-monotone}) if for all formulas $\varphi$,  valuations $\alpha$ and $\beta$ such that $\alpha\leq \beta$, $I,\alpha \models_w \Theta\,\varphi$ entails that
$I,\beta \models_w \Theta\,\varphi$ (resp., $I,\beta \models_w \Theta\,\varphi$ entails that
$I,\alpha \models_w \Theta\,\varphi$). By construction, all the parameterized operators are monotone. In particular, being $\PL$ and $\PU$
disjunct, by increasing
(resp., decreasing) the values of upward (resp., downward) parameters, the satisfaction relation is preserved.

\begin{proposition}\label{prop:MonotonicityPHS}
\begin{compactitem}
  \item The operators in $\PHS$ parameterized by variables in $\PU$ are upward-closed,
while those parameterized by variables in $\PL$ are downward-closed.
\item Let $\varphi$ be a $\PHS$ formula and let $\alpha$ and $\beta$ be variable valuations satisfying $\beta(u)\geq \alpha(u)$ for every $u\in \PU$ and $\beta (\ell) \leq \alpha(\ell)$ for every $\ell\in \PL$. Then  $(w,\alpha) \models \varphi$ entails
       that $(w,\beta) \models \varphi$.
\end{compactitem}
\end{proposition}

  Note that if we also allow for all $\ell\in \PL$ and $u\in \PU$,
the parameterized modalities $\hsX_{\succ  u}$, $\hsX_{\prec   \ell}$, $\hsUX_{\succ  \ell}$, and
$\hsUX_{\prec   u}$, then the modalities $\hsX_{\sim}\wp$ and $\hsUX_{\sim}\wp$, for $\wp\in \PU\cup \PL$ and $\sim\in \{<,\leq,>,\geq\}$ are dual and have opposite kind of monotonicity. It easily follows that the logic is indeed closed under negation.

\begin{proposition} Given a $\PHS$ formula $\varphi$ with upward (resp., downward) parameters in $\PU$ (resp., $\PL$), one can construct
in linear time a $\PHS$ formula $\overline{\varphi}$ with upward (resp., downward) parameters in $\PL$ (resp., $\PU$) corresponding to the negation of $\varphi$, i.e.~such that for each parameter valuation $\alpha$ and trace $w$ over $2^{\Prop}$,
$(w,\alpha)\models \varphi$ iff $(w,\alpha)\not\models \overline{\varphi}$.
\end{proposition}

We now show that parametric $\LTL$ ($\PLTL$)~\cite{AlurETP01}  can be easily expressed in  $\PHSF{\AB}$. Recall that  $\PLTL$ formulas $\varphi$ over $\Prop$ and the set of parameters $\PU\cup \PL$ are defined as:
\[
\varphi ::= \top \;\vert\;   p \;\vert\; \neg p \;\vert\; \varphi \vee \varphi \;\vert\; \varphi \wedge \varphi \;\vert\; \Next  \varphi  \;\vert\; \varphi \until \varphi \;\vert\;  \Always \varphi \;\vert\;  \Eventually_{\leq   u} \varphi \;\vert\; \Always_{\leq   \ell} \varphi
\]
where $p\in \Prop$, $u\in \PU$, $\ell\in \PL$, $\Next$, $\until$, and $\Always$ are the standard next, until, and always modalities, respectively, and $\Eventually_{\leq   u}$ and $\Always_{\leq   \ell}$  are parameterized versions of the always and eventually modalities. Other parameterized modalities such as $\Eventually_{> \ell}$ or $\Always_{>   u}$ can be easily expressed in the considered logic~\cite{AlurETP01}. For a $\PLTL$ formula $\varphi$,  a trace $w$, a parameter valuation $\alpha$, and a position $i\geq 0$, the satisfaction relation $(w,i,\alpha)\models \varphi$ is defined by induction as follows (we omit the  semantics of $\LTL$ constructs which is standard):
 \[
 \begin{array}{ll}
(w,i, \alpha) \models \Eventually_{\leq u}\varphi  &  \Leftrightarrow  \text{there is some $k\leq \alpha(u)$  such that
 }   (w,i+k, \alpha) \models \varphi;\\
(w,i, \alpha) \models \Always_{\leq \ell}\varphi  &  \Leftrightarrow  \text{for each $k\leq \alpha(\ell)$: }   (w,i+k, \alpha) \models \varphi.
\end{array}
\]

\begin{proposition}\label{prop:FromPLTLtoPAB}
For a $\PLTL$ formula $\varphi$, one can build in linear time a $\PHSF{\AB}$ formula $f(\varphi)$ such that for all traces
$w$, $i\geq 0$, and parameter valuations $\alpha$, $(w,i,\alpha)\models \varphi$ iff $([i,i],\alpha)\models_w f(\varphi)$.
\end{proposition}
\begin{proof}
The mapping $f : \PLTL \mapsto \PHSF{\AB}$, homomorphic with respect to atomic propositions and Boolean connectives, is
defined as follows:\vspace{-0.2cm}
\[
 \begin{array}{ll}
f(\Next\varphi)  &  \DefinedAs  \hsA(\len{2}\wedge \hsA (\len{1}\wedge \varphi));\\
f(\varphi \until\varphi_2) & \DefinedAs \hsA\bigl(\hsA(\len{1}\wedge f(\varphi_2))\wedge \hsUB\hsA(\len{1}\wedge f(\varphi_1))\bigr);\\
f(\Always\varphi)  &  \DefinedAs  \hsUA \hsA (\len{1}\wedge \varphi) ;\\
f(\Eventually_{\leq u}\varphi)  &  \DefinedAs  \hsA_{\leq u} \hsA (\len{1}\wedge \varphi); \\
f(\Always_{\leq \ell}\varphi)  &  \DefinedAs  \hsUA_{\leq \ell} \hsA (\len{1}\wedge \varphi).\vspace{-0.6cm}
\end{array}
\]
\end{proof}
Note that by Proposition~\ref{prop:FromPLTLtoPAB} and the results in~\cite{AlurETP01}, the relaxation of the assumption $\PU\cap \PL=\emptyset$ or the adding of parameterized operators of the form $\hsX_{=\wp}$ would lead to an undecidable model-checking problem already for the parameterized extension of $\AB$ by just one parameter.\vspace{0.2cm}

\noindent \textbf{Expressively complete fragments.} Two $\PHS$ formulas $\varphi$ and $\psi$ are \emph{strongly equivalent}, denoted by $\varphi\equiv\psi$, if for all traces, intervals $I$ over $\Nat$, and parameter valuations $\alpha$, we have that $(I,\alpha)\models_w \varphi$ iff $(I,\alpha)\models_w \psi$.
We show that the fragment consisting of  $\PHSF{\BBbarBbarWEEbarEbarW}$ formulas with no occurrences of parameterized operators
  $\hsUX_{\succ u}$ is sufficient to capture the full logic $\PHS$.

\begin{proposition}\label{prop:FragmentCompleteForPHS} Given a $\PHS$ formula $\varphi$, one can build in linear time
a strongly equivalent $\PHSF{\BBbarBbarWEEbarEbarW}$ formula $\psi$ with no occurrences of the parameterized operators
$\hsUX_{\succ u}$.
\end{proposition}
\begin{proof} We first show that the fragment $\PHSF{\BBbarBbarWEEbarEbarW}$ is expressively complete for $\PHS$.
The strong equivalences exploited for expressing all the $\HS$ modalities in terms of the modalities in
the fragment  $\BBbarEEbar$ can be trivially adapted to the parameterized setting. Here, we illustrate the equivalences for the existential parameterized operators where $\sim\in\{<,\leq,>,\geq\}$ and $\wp\in \PU\cup \PL$:
 \[
 \begin{array}{lll}
 \hsA_{\sim  \wp} \varphi  &  \equiv & (\len{1}\wedge \hsBt_{\sim  \wp}\varphi)\vee \hsE (\len{1}\wedge \hsBt_{\sim  \wp}\varphi);\\
  \hsAt_{\sim  \wp} \varphi  &  \equiv & (\len{1}\wedge \hsEt_{\sim  \wp}\varphi) \vee \hsB (\len{1}\wedge \hsEt_{\sim  \wp}\varphi);\\
  \hsL_{\sim  \wp} \varphi  &  \equiv & \hsBt\hsE (\len{1}\wedge \hsBt_{\sim  \wp}\varphi);\\
  \hsLt_{\sim  \wp} \varphi  &  \equiv & \hsEt\hsB (\len{1}\wedge \hsEt_{\sim  \wp}\varphi);\\
   \hsD_{\sim  \wp} \varphi  &  \equiv & \hsB\hsE_{\sim \wp}\varphi;\\
  \hsDt_{\sim  \wp} \varphi  &  \equiv & \hsBt\hsEt_{\sim \wp}\varphi;\\
  \hsO_{\sim  \wp} \varphi  &  \equiv & \hsE (\neg\len{1}\wedge \hsBt_{\sim  \wp}\varphi);\\
    \hsOt_{\sim  \wp} \varphi  &  \equiv & \hsB (\neg\len{1}\wedge \hsEt_{\sim  \wp}\varphi).
\end{array}
\]
 It remains to show that for the fragment $\PHSF{\BBbarBbarWEEbarEbarW}$, the universal upward parameterized operators can be expressed in terms of the other modalities.  One can easily show that the following strong equivalences
  hold, where $\prec$ is $<$ (resp., $\prec$ is $\leq$) and $\succ$ is $\geq$ (resp., $\succ$ is $>$). Hence, the result follows.
 \[
 \begin{array}{lll}
 \hsUB_{\succ  u} \varphi  &  \equiv & \hsUB (  \varphi \vee \hsBtW_{\prec u}\top );\\
 \hsUE_{\succ  u} \varphi  &  \equiv & \hsUE (  \varphi \vee\hsEtW_{\prec u}\top );\\
 \hsUBt_{\succ  u} \varphi  &  \equiv & \hsBt_{\prec u}(\hsUBt \varphi) \vee  \hsUBt \varphi;\\
 \hsUEt_{\succ  u} \varphi  &  \equiv & \hsEt_{\prec u}(\hsUEt \varphi) \vee  \hsUEt \varphi; \\
 \hsUBtW_{\succ  u} \varphi  &  \equiv & \hsBtW_{\prec u}(\hsUBt \varphi) \vee  \hsUBtW \varphi; \\
 \hsUEtW_{\succ  u} \varphi  &  \equiv & \hsEtW_{\prec u}(\hsUEt \varphi) \vee  \hsUEtW \varphi.\vspace{-0.6cm}
\end{array}
\]
\end{proof}

For the logic $\PHSF{\ABBbarBbarW}$, we obtain a similar result.

\begin{proposition}\label{prop:FragmentCompleteForPABBbar} Given a $\PHSF{\ABBbarBbarW}$ formula $\varphi$, one can build in linear time
a strongly equivalent $\PHSF{\ABBbarBbarW}$ formula $\psi$ with no occurrences of the parameterized operators
$\hsUX_{\succ u}$.
\end{proposition}
\begin{proof}
The result directly follows from the strong equivalences provided in the proof of Proposition~\ref{prop:FragmentCompleteForPHS}
and the following one, where $\prec$ is $<$ (resp., $\prec$ is $\leq$) and $\succ$ is $\geq$ (resp., $\succ$ is $>$):
  $\hsUA_{\succ  u} \varphi   \,\,  \equiv \,\, \hsA_{\prec  u}\hsUBt\varphi. $
\end{proof}

By the monotonicity of the parameterized modalities and Propositions~\ref{prop:FragmentCompleteForPHS} and~\ref{prop:FragmentCompleteForPABBbar}, we can eliminate all the parameterized modalities, but the existential upward ones,
for solving the model-checking and satisfiability problems against $\PHS$ (resp., $\PHSF{\ABBbarBbarW}$).

\begin{lemma}\label{lemma:ReductionToPrompness} Model checking $\PHS$ (resp., $\PHSF{\ABBbarBbarW}$)
can be reduced in linear time to model checking $\PromptHS$ (resp., $\PromptF{\ABBbarBbarW}$). Similarly,
satisfiability of $\PHS$ (resp., $\PHSF{\ABBbarBbarW}$)
can be reduced in linear time to satisfiability of $\PromptHS$ (resp., $\PromptF{\ABBbarBbarW}$).
\end{lemma}
\begin{proof} Let $\varphi$ be a $\PHS$ (resp., $\PHSF{\ABBbarBbarW}$) formula. By Propositions~\ref{prop:FragmentCompleteForPHS} and~\ref{prop:FragmentCompleteForPABBbar}, we can assume that $\varphi$ does not contain occurrences of parameterized operators
of the form $\hsUX_{\succ u}$. Let $f(\varphi)$ be the $\PromptHS$ (resp., $\PromptF{\ABBbarBbarW}$) formula intuitively obtained from $\varphi$ by replacing each occurrence of a downward parameter $\ell$ with the constant $1$. Formally, $f(\varphi)$ is homomorphic w.r.t.~all the constructs but the downward parameterized modalities and:
\begin{compactitem}
  \item $f(\hsX_{\geq \ell} \varphi)\DefinedAs \hsX  f(\varphi)$;
  \item $f(\hsX_{> \ell} \varphi)\DefinedAs \hsX(\neg\len{1}\wedge f(\varphi))$;
  \item $f(\hsUX_{\leq \ell}\varphi)\DefinedAs \hsUX (\neg \len{1}\vee f(\varphi))$;
  \item $f(\hsUX_{<\ell}\varphi)\DefinedAs  \top$.
\end{compactitem}
As for the model checking problem, we show that $V(\Ku,\varphi)\neq \emptyset$
iff $V(\Ku,f(\varphi))\neq \emptyset$ for each Kripke structure $\Ku$. Let $\alpha_1$ be a parameter valuation such that $\alpha_1(\ell)=1$ for each downward parameter
$\ell\in L$.
By construction,   
for all traces $w$,
$(w,\alpha_1)\models  \varphi$ iff $(w,\alpha_1)\models f(\varphi)$.
Hence, $V(\Ku,f(\varphi))\neq \emptyset$
implies that  $V(\Ku,\varphi)\neq \emptyset$. On the other hand, if $V(\Ku,\varphi)\neq \emptyset$, there is a parameter valuation
$\alpha$ such that for each trace $w$ of $\Ku$, $(w,\alpha)\models \varphi$. Let $\alpha_1$ be defined as: $\alpha_1(u)=\alpha(u)$
for each $u\in U$, and $\alpha_1(\ell)= 1$ for each $\ell\in L$. By Proposition~\ref{prop:MonotonicityPHS}, it follows that  for each trace $w$ of $\Ku$, $(w,\alpha_1)\models \varphi$. Thus,  we obtain that $V(\Ku,f(\varphi))\neq \emptyset$ as well, and the result for the model-checking problem follows. The result for the satisfiability problem is similar.
\end{proof}

\section{Decision procedures for $\PHS$}\label{sec:DecisionProcedurePHS}

In this section, we first provide a translation of $\HS$  formulas into equivalent B\"{u}chi $\NFA$ (asymptotically optimal for $\ABBbarBbarW$ formulas), by exploiting as an intermediate step a translation of $\HS$  formulas into equivalent formulas of \emph{linear-time hybrid logic $\HL$}~\cite{FRS03,SW07,BozzelliL10} (Subsection~\ref{sec:LinearTimeHybridLogic}). Then, in Subsection~\ref{sec:SatisfiabilityModelCheckingPHS}, we apply the results of  Subsection~\ref{sec:LinearTimeHybridLogic} and the alternating color technique
for Prompt $\LTL$~\cite{KupfermanPV09} in order to solve
satisfiability and model checking against $\PHS$ and $\PHSF{\ABBbarBbarW}$. In particular, for the logic
$\PHSF{\ABBbarBbarW}$, we show that the considered problems are \EXPSPACE-complete.

\subsection{Translation of $\HS$ in linear-time Hybrid Logic}\label{sec:LinearTimeHybridLogic}

In this section,  we recall the  linear-time hybrid logic $\HL$~\cite{FRS03,SW07,BozzelliL10}, which extends standard  $\LTL$ + past  by first-order concepts. We show that while $\HS$ can be translated into the two-variable fragment of $\HL$, for the logic $\ABBbarBbarW$,  it suffices to consider the one-variable fragment $\HL_1$ of $\HL$. Thus, by exploiting known results on $\HL_1$~\cite{SW07,BozzelliL10}, we obtain an asymptotically optimal automata-theoretic approach for $\ABBbarBbarW$ of elementary complexity.  \vspace{0.1cm}

 \noindent \textbf{Syntax and semantics of $\HL$.} Given  a  set $X$ of (position) variables,
the set of $\HL$ formulas $\varphi$ over $\Prop$ and $X$ is defined by the following syntax:
\[
\varphi \DefinedAs \top \ |\ p  \ |\  x  \ |\ \neg\,\varphi \ |\ \varphi\, \wedge\, \varphi \ |\
  \Next  \varphi \ |\ \PNext  \varphi \ |\
  \Eventually \varphi\ |\ \PEventually \varphi  \ |\ \DBinder. \varphi
\]
\noindent   $p \in \Prop$, $x\in X$, $\PNext$ and   $\PEventually$  are the past counterparts of the next modality $\Next$ and
the eventually modality $\Eventually$, respectively,  and  $\DBinder$ is the \emph{downarrow binder} operator which assigns the variable name $x$ to the current position. We denote by $\HL_1$ (resp., $\HL_2$) the one-variable (resp., two-variable) fragment of
$\HL$.  An $\HL$ \emph{sentence} is a formula where each variable $x$ is not free (i.e., occurs in the scope of a binder modality $\DBinder$). The size $|\varphi|$ of an $\HL$ formula $\varphi$ is the number of distinct subformulas of $\varphi$.

$\HL$  is interpreted over traces $w$. A \emph{valuation} $g$ is a mapping  assigning to each variable  a position $i\geq 0$. The satisfaction relation
$(w,i,g)\models \varphi$, meaning that $\varphi$ holds at position $i$ along $w$ w.r.t. the valuation $g$, is inductively  defined as follows (we omit the semantics of $\LTL$ constructs which is standard):
\[
\begin{array}{ll}
 (w,i,g)\models x &  \Leftrightarrow i=g(x)\\
 (w,i,g)\models  \DBinder.\varphi & \Leftrightarrow
                 (w,i,g[x \mapsto i])\models \varphi
\end{array}
\]

\noindent where $g[x\mapsto i](x)=i$ and $g[x\mapsto i](y)=g(y)$ for $y\neq x$. Thus,  $\DBinder$  binds the variable $x$ to the current position. Note that the satisfaction relation  depends only on the values assigned to the variables occurring free  in the given formula $\varphi$. We write $(w,i)\models \varphi$ to mean that
$(w,i,g_0)\models \varphi$, where $g_0$ maps each variable to position 0, and $w\models \varphi$ to mean that $(w,0)\models \varphi$.   Note that
  $\HL$ formulas can be trivially translated into equivalent formulas of first-order logic $\FO$ over traces and $\LTL$ formulas can be trivially translated into equivalent $\HL$ formulas. Thus, by the first-order expressiveness completeness of $\LTL$, $\HL$ and $\LTL$ have the same expressiveness~\cite{FRS03}.\vspace{0.1cm}

 \noindent \textbf{Translation of $\HS$ into $\HL$.} We establish the following result.

 \begin{proposition}\label{prop:FromHStoHL}
 Given an $\HS$ (resp., $\ABBbarBbarW$) formula $\varphi$, one can construct in linear-time a two-variable (resp., one-variable)
 sentence $\HL$ $\varphi'$ such that $\Lang(\varphi)=\Lang(\varphi')$.
 \end{proposition}
 \begin{proof} We first consider full $\HS$. We can restrict ourselves to consider the fragment $\BBbarEEbar$ of $\HS$ since
 all temporal modalities in $\HS$ can be expressed in $\BBbarEEbar$  by a linear-time translation. Fix two distinct variables $x_L$ and $x_R$. We define a mapping
 $f: \BBbarEEbar \mapsto \HL_2$ assigning to each $\BBbarEEbar$ formula $\varphi$ a $\HL_2$ formula $f(\varphi)$ with variables $x_L$ and $x_R$ which occur free in $f(\varphi)$. Intuitively, in the translation, $x_L$ and $x_R$ refer to the left and right endpoints of the current interval in $\Nat$, while the current position corresponds to the left endpoint of the current interval.
  Formally,  the mapping $f$ is homomorphic w.r.t.~the Boolean connectives and is inductively defined as follows:

  \begin{compactitem}
  \item $f(p)\DefinedAs \Always(\Eventually x_r \rightarrow p)$;
  \item $f(\hsB\varphi)     \DefinedAs  \Eventually (\Next\Eventually x_R \wedge \DBinder_R.\, \PEventually (x_L\wedge f(\varphi)))$;
    \item $f(\hsBt\varphi)    \DefinedAs  \Eventually (x_R \wedge \Next\Eventually \DBinder_R.\, \PEventually (x_L\wedge f(\varphi)))$;
  \item $f(\hsE\varphi)     \DefinedAs  \Next\Eventually (\Eventually x_R \wedge \DBinder_L.\, f(\varphi))$;
   \item $f(\hsEt\varphi)     \DefinedAs  \Next\PEventually  \DBinder_L.\, f(\varphi)$.
 \end{compactitem}

 By a straightforward induction on $\varphi$, we obtain that given a trace $w$, an interval $[i,j]$, a valuation $g$ such that
 $g(x_L)=i$ and $g(x_R)=j$, it holds that $[i,j]\models_w \varphi$ if and only if $(w,i,g)\models f(\varphi)$. The desired
 $\HL_2$ sentence $\varphi'$ equivalent to $\varphi$ is then defined as follows: $\varphi'\DefinedAs \DBinder_L.\, \DBinder_R. \,f(\varphi)$.\vspace{0.1cm}

  We now consider the logic  $\ABBbarBbarW$. We can restrict ourselves to consider the fragment $\ABBbar$ of $\HS$ since
the modality $\hsBtW$ can be trivially expressed in terms of $\hsBt$.   Fix a variable  $x$. We define a mapping
 $h: \ABBbar \mapsto \HL_1$ assigning to each $\ABBbar$ formula $\varphi$ an $\HL_1$ formula $h(\varphi)$ with one variable  $x$,  which occurs free in $h(\varphi)$. Intuitively, in the translation, $x$   refers to the left  endpoint  of the current interval in $\Nat$, while the current position corresponds to the right endpoint of the current interval.
  Formally,  the mapping $h$ is homomorphic w.r.t.~the Boolean connectives and is inductively defined as follows:

  \begin{compactitem}
  \item $h(p)\DefinedAs \neg\PEventually(\PEventually x  \wedge \neg p)$;
  \item $h(\hsA) \DefinedAs \DBinder.\, \Eventually h(\varphi)$;
 \item $h(\hsB\varphi)     \DefinedAs  \PNext\PEventually (h(\varphi) \wedge \PEventually x)$;
  \item $h(\hsBt\varphi)    \DefinedAs  \Next\Eventually  h(\varphi)$.
 \end{compactitem}

 By a straightforward induction on $\varphi$, we can prove that given a trace $w$, an interval $[i,j]$, a valuation $g$ such that
 $g(x)=i$, it holds that $[i,j]\models_w \varphi$ if and only if $(w,j,g)\models h(\varphi)$. The desired
 $\HL_1$ sentence $\varphi'$ equivalent to $\varphi$ is then defined as follows: $\varphi'\DefinedAs \DBinder. \,h(\varphi)$.
 \end{proof}

 It is known that  $\HL_2$ is already non-elementarily decidable~\cite{SW07} and for an $\HL$ formula $\varphi$, one can construct
 a B\"{u}chi $\NFA$ accepting $\Lang(\varphi)$ whose size is a tower of exponentials having height equal to the nesting depth of the binder modality plus one~\cite{BozzelliL10}. For the one-variable fragment $\HL_1$ of $\HL$, one can do much better~\cite{BozzelliL10}: the size of the B\"{u}chi $\NFA$ equivalent to a $\HL_1$ formula $\varphi$ has size doubly exponential in the size of $\varphi$.
 Hence, by Proposition~\ref{prop:FromHStoHL}, we obtain the following result.

 \begin{proposition}\label{prop:FromHStoBuchiNFA}
 Given an $\HS$ formula $\varphi$, one can build a B\"{u}chi $\NFA$ $\Au_\varphi$ accepting $\Lang(\varphi)$. Moreover, if
 $\varphi$ is a  $\ABBbarBbarW$ formula, then $\Au_\varphi$ has size doubly exponential in the size of $\varphi$.
 \end{proposition}

 Note that by~\cite{BozzelliL10}, the B\"{u}chi $\NFA$ equivalent to a $\HL_1$ formula can be built on the fly. Recall that non-emptiness of  B\"{u}chi $\NFA$ is
 \NLOGSPACE-complete, and  the standard model checking algorithm consists in checking emptiness of the B\"{u}chi $\NFA$ resulting from the synchronous product of the given finite Kripke structure with the B\"{u}chi $\NFA$ associated with the negation of the fixed formula. Thus, by Proposition~\ref{prop:FromHStoBuchiNFA}, we obtain algorithms for satisfiability and model-checking of $\ABBbarBbarW$ which run in non-deterministic single exponential space. In~\cite{BozMPS2021}, it is shown that satisfiability and model checking of $\AB$ over \emph{finite words} is already \EXPSPACE-hard. The \EXPSPACE-hardness proof in~\cite{BozMPS2021} can be trivially adapted to handle $\AB$ over infinite words. Thus, since \EXPSPACE\  $=$ \NEXPSPACE, we obtain the following result.

 \begin{corollary}Model checking and satisfiability problems for $\ABBbarBbarW$ are both \EXPSPACE-complete.
 \end{corollary}

 \subsection{Solving satisfiability and model checking of $\PHS$}\label{sec:SatisfiabilityModelCheckingPHS}

In this section, we provide an automata-theoretic approach for solving satisfiability and model checking of
$\PromptHS$ and $\PromptF{\ABBbarBbarW}$ based on Proposition~\ref{prop:FromHStoBuchiNFA} and the alternating color technique for Prompt $\LTL$~\cite{KupfermanPV09}. By Lemma~\ref{lemma:ReductionToPrompness},
we devise algorithms for solving satisfiability and model checking against $\PHS$ and $\PHSF{\ABBbarBbarW}$ as well, which for the case of $\PHSF{\ABBbarBbarW}$ are asymptotically optimal.\vspace{0.2cm}

\noindent \textbf{Alternating color technique~\cite{KupfermanPV09}.} We fix a fresh proposition $c\notin \Prop$. Let us consider
a trace $w$. A \emph{$c$-coloring of $w$} is a trace $w'$ over $\Prop\cup\{c\}$ such that $w$ and $w'$ agree at every position on all the truth values of the propositions in $\Prop$, i.e.~$w'(i)\cap\Prop=w(i)$ for all $i\geq 0$. A position $i\geq 0$ is a
\emph{$c$-change point} in $w'$ if \emph{either} $i=0$, \emph{or} the colors of $i$ and $i-1$ are different,
i.e.~$c\in w'(i)$ iff $c\notin w'(i-1)$.
A \emph{$c$-block of $w'$} is a maximal interval $[i,j]$ which has exactly one $c$-change point in $w'$, and this change point is at the first position $i$ of  $[i,j]$. Given $k\geq 1$, we say that $w'$ is \emph{$k$-bounded} if each $c$-block of $w'$ has length at most $k$,
which implies that $w'$ has infinitely many $c$-change points. Dually, we say that $w'$ is
\emph{$k$-spaced} if $w'$ has infinitely many $c$-change points and every $c$-block has length at least $k$.

We apply the alternating color technique~\cite{KupfermanPV09} for replacing a parameterized modality $\hsX_{\prec u}\psi$ in $\PromptHS$ with a  non-parameterized one requiring that the selected interval where $\psi$ holds has at most one $c$-change point.
Formally, let $\Rel_c: \PromptHS \mapsto \HS$ be the mapping associating to each $\PromptHS$ formula a $\HS$ formula, homomorphic
w.r.t.~propositions, connectives, and non-parameterized modalities, and defined as follows on parameterized formulas $\hsX_{\prec u}\psi$:
\[
\Rel_c(\hsX_{\prec u}\psi)\DefinedAs \hsX(\Rel_c(\psi) \wedge (\theta_c\vee \theta_{\neg c})).
\]
where for each $d\in\{c,\neg c\}$, $\theta_d$ is an $\AB$ formula requiring that the current interval has at most one $c$-change point and the right endpoint is a $d$-colored position:
\[
\theta_d\DefinedAs \hsA(\len{1}\wedge d)\wedge \hsUB(\hsA(\len{1}\wedge \neg d) \rightarrow \hsUB \hsA(\len{1}\wedge \neg d) ).
\]
For a $\PromptHS$ formula $\varphi$, let  $c(\varphi)$ be the $\HS$ formula defined as follows:
\[
c(\varphi)\DefinedAs \Rel_c(\varphi)\wedge \Alt_c \quad\quad \Alt_c\DefinedAs \hsUA\hsA\hsA(\len{1}\wedge c)\wedge \hsUA\hsA\hsA(\len{1}\wedge \neg c)
\]
Note that $c(\varphi)$ is a $\ABBbarBbarW$ formula if $\varphi$ is a $\PromptF{\ABBbarBbarW}$ formula. Moreover, the $\AB$ formula
$\Alt_c$ requires that there are infinitely many $c$-change points. Thus, $c(\varphi)$ forces a $c$-coloring of the given trace $w$  to be partitioned into infinitely many blocks such that each parameterized modality selects an interval with at most one $c$-change point. Like
Prompt $\LTL$~\cite{KupfermanPV09}, there is a weak equivalence between $\varphi$ and $c(\varphi)$ on $k$-bounded and $k$-spaced $c$-coloring of $w$. The following lemma rephrases Lemma 2.1 in~\cite{KupfermanPV09} and can be proved in a similar way.

\newcounter{lemma-AlternatingColororing}
\setcounter{lemma-AlternatingColororing}{\value{lemma}}

\begin{lemma}\label{lemma:AlternatingColororing} Let $\varphi$ be a $\PromptHS$ formula and $w$ be a trace.
\begin{compactenum}
  \item If $(w,\alpha)\models \varphi$, then $w'\models c(\varphi)$ for each $k$-spaced $c$-coloring $w'$ of $w$ with $k= \max_{u\in \PU}\alpha(u)$.
  \item Let $k\geq 1$. If $w'$ is a $k$-bounded $c$-coloring of $w$ with $w'\models c(\varphi)$, then $(w,\alpha)\models \varphi$,
  where $\alpha(u)=2k$ for each $u\in \PU$.
\end{compactenum}
\end{lemma}

  \noindent \textbf{Solving satisfiability.} Let $\varphi$ be a $\PromptHS$ formula and $\Au_c$ be the B\"{u}chi $\NFA$
 of Proposition~\ref{prop:FromHStoBuchiNFA} accepting the models of the $\HS$  formula $c(\varphi)$. By Lemma~\ref{lemma:AlternatingColororing}, we deduce that $S(\varphi)\neq \emptyset$ if and only if there is $k\geq 1$ and some $k$-bounded $c$-coloring $w'$ accepted by $\Au_c$. Indeed, if $S(\varphi)\neq\emptyset$, then there is a parameter valuation $\alpha$  and a trace $w$ such that $(w,\alpha)\models \varphi$. Let $k= \max_{u\in \PU}\alpha(u)$ and $w'$ be the $c$-coloring
 of $w$ whose $c$-blocks have length exactly $k$. Note that $w'$ is both $k$-spaced and $k$-bounded. By Lemma~\ref{lemma:AlternatingColororing}(1), $w'\models c(\varphi)$, hence, $w'$ is accepted by $\Au_c$. Vice versa, if there is a trace $w$ and  a $k$-bounded $c$-coloring $w'$ of $w$ accepted by $\Au_c$, then, by Lemma~\ref{lemma:AlternatingColororing}(2),
 $(w,\alpha)\models \varphi$, where $\alpha(u)= 2k$  for each $u\in \PU$. Hence, $S(\varphi)\neq \emptyset$.

 Let $N_c$ be the number of  $\Au_c$ states. Assume that there is  $k$-bounded $c$-coloring $w'$ accepted by $\Au_c$ for some $k\geq 1$. We claim that there is also a $2N_C+1$-bounded $c$-coloring accepted by $\Au_c$. If $k\leq 2N_c+1$, the result is obvious. Otherwise, let $\pi$ be an accepting run of $\Au_c$ over $w'$, and let us consider the infixes $\nu$ of $\pi$ associated with the $c$-blocks of $w'$ greater than $2N_C+1$. We replace  $\nu$ with an infix of length at
 most $2N_C+1$ as follows:
 \begin{itemize}
   \item If $\nu$ does not visits accepting states, we remove from $\nu$ the maximal cycles (but the first states of such cycles) by obtaining a finite path of length at most $N_c$.
   \item If $\nu$ visits some accepting state, then $\nu$ can be written in the form $\nu=\nu_1\cdot q_{a}\cdot \nu_2$, where $q_a$ is an accepting state. We remove the maximal cycles from $\nu_1$ and $\nu_2$ (but the first states of such cycles) by obtaining a finite path of length at most $2N_c+1$.
 \end{itemize}

 In this way, we obtain an accepting run of $\Au_c$ over a $2N_c+1$-bounded $c$-coloring, and the result follows. Then,
 starting from $\Au_c$, one can easily construct in time polynomial in the size of $\Au_c$, a B\"{u}chi $\NFA$ $\Au'_c$ accepting the  $2N_c+1$-bounded colorings which are accepted by $\Au_c$. $\Au'_c$ keeps track in its state of the current state of $\Au_c$ and the binary encoding of the value of a counter modulo $2N_c+1$, where the latter is reset whenever a $c$-change point occurs. Note that if $\varphi$ is a $\PromptF{\ABBbarBbarW}$ formula, then $c(\varphi)$ is
 a $\ABBbarBbarW$ formula, and by Proposition~\ref{prop:FromHStoBuchiNFA}, the size of $\Au_c$ is doubly exponential in the size of $\varphi$. By the previous observations, it follows that
 if $S(\varphi)\neq \emptyset$, then there is a parameter valuation $\alpha\in S(\varphi)$ which is bounded doubly exponentially in $|\varphi|$.
 Thus, since non-emptiness of B\"{u}chi $\NFA$ is \NLOGSPACE-complete, by Lemma~\ref{lemma:ReductionToPrompness} and Proposition~\ref{prop:FromHStoBuchiNFA}, we obtain the following result.

  \begin{theorem}\label{theo:PHSsatisfiability} Satisfiability of $\PHS$ is decidable. Moreover, satisfiability of
$\PHSF{\ABBbarBbarW}$ is \EXPSPACE-complete and given a $\PHSF{\ABBbarBbarW}$ formula $\varphi$, in case $S(\varphi)\neq\emptyset$,
  there is a parameter valuation in  $S(\varphi)$ which is bounded doubly exponentially in $|\varphi|$.
  \end{theorem}

  We now show that the double exponential upper bound on the values of the parameters in Theorem~\ref{theo:PHSsatisfiability} for
  satisfiable $\PHSF{\ABBbarBbarW}$ formulas cannot be in general improved. Indeed, we provide a matching lower bound by defining for each $n\geq 1$, a $\PHSF{\AB}$ formula of size polynomial in $n$ which encodes a yardstick of length
  $(n+1)*2^{n}* 2^{2^{n}}$. This is done by using a $2^{n}$-bit counter for expressing integers in the range
$[0,2^{2^{n}}-1]$ and an $n$-bit counter for keeping track of the position (index) $i\in [0,2^{n}-1]$ of the $(i+1)^{th}$-bit of each valuation $v$ of the $2^{n}$-bit counter. In particular, such a valuation $v\in [0,2^{2^{n}}-1]$ is encoded by a sequence, called \emph{$n$-block}, of $2^{n}$ sub-blocks of length $n+1$ where
for each $i\in [0,2^{n}-1]$, the $(i+1)^{th}$ sub-block encodes both the value and the index of the $(i+1)^{th}$-bit in the binary representation of $v$.

Formally, let $\Prop\DefinedAs \{\#_1,\#_2,\$,0,1 \}$. Fix $n\geq 1$. An \emph{$n$-sub-block} is a finite word $\nu$ over $2^{\Prop}$ of length $n+1$ of the form $\nu =\{\#_1,p,bit\}\{bit_1\},\ldots,\{bit_n\}$ where $bit,bit_1,\ldots,bit_n\in \{0,1\}$ and $p\in \{\#_1,\#_2\}$. If $p= \#_2$, we say that $\nu$ is marked. The \emph{content} of $\nu$ is $bit$, and the \emph{index} of $\nu$ is the number in $[0,2^{n}-1]$ whose binary code is
$bit_1,\ldots,bit_n$. An \emph{$n$-block} is a finite word $\nu$ of length $(n+1)*2^{n}$ of the
form $\nu=\nu_0\ldots \nu_{2^{n}-1}$, where $\nu_0$ is a marked $n$-sub-block of index $0$, and for each $i\in [1,2_{n}-1]$, $\nu_i$
is an unmarked $n$-sub-block having index $i$. The index of $i$ is the natural number
in $[0,2^{2^{n}}-1]$ whose binary code is $bit_0,\ldots,bit_{2^{n}-1}$, where $bit_i$ is the content of the sub-block $\nu_i$
for each $i\in [0,2^{n}-1]$. The yardstick of length $(n+1)*2^{n}* 2^{2^{n}}$ is then encoded by the trace, called \emph{$n$-trace},
given by
%
$bl_0\cdot\ldots \cdot bl_{2^{2^{n}}-1}\cdot \{\$\}^{\omega}$
%
where $bl_i$ is the $n$-block having index $i$ for each $i\in [0,2^{2^{n}}-1]$. We first show the following result.

\newcounter{lemma-lowerBoundParameters}
\setcounter{lemma-lowerBoundParameters}{\value{lemma}}

\begin{lemma}\label{lemma:lowerBoundParameters} For each $n\geq 1$, one can construct in time polynomial in $n$ a satisfiable
$\AB$ formula $\psi_n$ whose unique model is the $n$-trace.
\end{lemma}
\begin{proof} Fix $n\geq 1$. The $\AB$ formula $\psi_n$ is defined as
%
$\psi_n\DefinedAs \psi_{bl}\wedge  \psi_{inc}.$
%
The conjunct $\psi_{bl}$ captures the traces $w$ over $\Prop = \{\#_1,\#_2,\$,0,1 \}$ having the form
%
$bl_0\cdot\ldots \cdot bl_{k}\cdot \{\$\}^{\omega}$, for some $k\geq 1$,
%
such that the following conditions are satisfied:
\begin{compactitem}
  \item  $bl_0,\ldots,bl_k$ are $n$-blocks;
  \item $bl_0$ is the $n$-block of index $0$ (i.e., each $n$-sub-block of $bl_0$ has content $0$);
  \item $bl_k$ is the $n$-block of index $2^{2^{n}}-1$ (i.e., each $n$-sub-block of $bl_0$ has content $1$).
\end{compactitem}
One can easily construct an $\LTL$ formula of size polynomial in $n$ characterizing the traces satisfying the previous requirements.
Thus, since an $\LTL$ formula can be translated in linear time into an equivalent $\AB$ formula~\cite{BozzelliMMPS19}, we omit the details of the construction of the $\AB$ formula $\psi_{bl}$.

The conjunct $\psi_{inc}$ additionally ensures that $k= 2^{2^{n}}-1$ and for each $i\in [1,2^{2^{n}}-2]$, $bl_i$ is the $n$-block of index $i$. To this purpose, it suffices to guarantee  that in moving from a non-last $n$-block
$bl$ to the next one $bl'$, the $2^{n}$-counter is incremented.
This is equivalent to require that there is an $n$-sub-block $sbl_0$ of
$bl$ whose content is $0$ such that for each $n$-sub-block $sbl$
of $bl$, denoted by $sbl'$ the $n$-sub-block of $bl'$ having the same
index as $sbl$, the following holds: (i) if $sbl$
precedes $sbl_0$, then the content of $sbl$ (resp., $sbl'$) is $1$
(resp., $0$), (ii) if $sbl$ corresponds to $sbl_0$, then the content
of $sbl'$ is $1$, and (iii) if $sbl$ follows $sbl_0$, then there is
$b\in\{0,1\}$ such that the content of both $sbl$ and $sbl'$ is $b$. In order to express these conditions,
we define
auxiliary $\AB$ formulas. Recall that proposition $\#_1$ marks the first position of an $n$-sub-block, while
$\#_2$ marks the first position of an $n$-block $bl$ (corresponding to the first position of the first $n$-sub-block of $bl$).  For each $\AB$ formula $\varphi$, the $\AB$ formula $right(\varphi)$ requires that $\varphi$ holds at the singleton interval corresponding to the right endpoint of the current interval:
\[
right(\varphi)\DefinedAs \hsA(\len{1} \wedge \varphi)
\]
The $\AB$ formula $\psi_{one}(\#_2)$ ensures that proposition $\#_2$ occurs exactly once in the current interval, while
$\psi_{not}(\#_2)$ ensures that $\#_2$ does not occur  in the current interval. We focus on the definition of
$\psi_{one}(\#_2)$ (the definition of $\psi_{not}(\#_2)$ being similar).
\[
\psi_{one}(\#_2)\DefinedAs [right(\#_2) \vee \hsB right(\#_2)] \wedge \neg  \hsB[right(\#_2) \wedge \hsB right(\#_2)]
\wedge \neg [right(\#_2) \wedge \hsB right(\#_2)]
\]
Moreover, we define the $\AB$ formulas
 $\psi_=(b,b')$ where $b,b'\in\{0,1\}$.
Formula $\psi_=(b,b')$ holds at a singleton interval $[h,h]$ (along the given trace) iff whenever $h$ corresponds to the beginning of a $n$-sub-block
$sbl$ of an $n$-block $bl$, then (i) the content of $sbl$ is $b$,
(ii) the $n$-block $bl$ is followed by an $n$-block $bl'$, and (iii)
the $n$-sub-block of $bl'$ having the same index as
$sbl$ has content $b'$.
\[
\begin{array}{l}
     \psi_=(b,b')   \,\, \DefinedAs \,\,  \#_1 \rightarrow \bigl[b \wedge \hsA \bigl(\len{2}\wedge \hsA( \psi_{one}(\#_2) \wedge \theta_= \wedge right(b'\wedge \#_1))\bigr)\bigr]\vspace{0.2cm}\\
     \theta_=    \,\, \DefinedAs \,\, \displaystyle{\bigwedge_{h=1}^{n} \bigvee_{b\in \{0,1\}}} \bigl[ \hsB(\len{h}\wedge right(b))
     \wedge \hsA(\len{h+1}\wedge right(b))\bigr]
\end{array}
\]
 Finally, the conjunct  $\psi_{inc}$ in the definition of $\psi_n$ is given by
\[
\begin{array}{l}
     \psi_{inc}  \,\, \DefinedAs \,\, \hsUA \bigl([right(\#_2)\wedge \hsA(\neg \len{1} \wedge right(\#_2))] \longrightarrow
     \hsA [\psi_{one}(\#_2)\wedge right(\#_1\wedge \psi_=(0,1)) \wedge \psi_L \wedge \psi_R] \bigr) \vspace{0.2cm}\\
     \psi_L \,\, \DefinedAs \,\, \hsUB (right(\#_1) \rightarrow right(\psi_=(1,0)))\vspace{0.2cm}\\
     \psi_R \,\, \DefinedAs \,\, \hsA \Bigl(\len{2} \wedge \hsUA \bigl[ (\psi_{not}(\#_2) \wedge right(\#_1)) \rightarrow \displaystyle{\bigvee_{b\in\{0,1\}}} right(\psi_=(b,b))  \bigr] \Bigr )
\end{array}
\]
 This concludes the proof of Lemma~\ref{lemma:lowerBoundParameters}.
\end{proof}
Fix $n\geq 1$ and let $\psi_n$ be the $\AB$ formula in Lemma~\ref{lemma:lowerBoundParameters}. We consider the $\PHSF{\AB}$
formula $\varphi_n$ with just one parameter given by $\varphi_n\DefinedAs \psi_n\wedge \hsA_{\leq u} \hsA(\len{1}\wedge\$)$.
By Lemma~\ref{lemma:lowerBoundParameters}, the smallest value for parameter $u$ for which  $\varphi_n$ has a model is greater
than  $(n+1)*2^{n}* 2^{2^{n}}$. Hence, we obtain the following result.

\begin{proposition} There is a finite set $\Prop$ of atomic propositions and  a family $\{\varphi_n\}_{n\geq 1}$ of satisfiable $\PHSF{\AB}$ formulas over $\Prop$ with just one parameter such that for each $n\geq 1$, $\varphi_n$ has size polynomial in $n$ and
the smallest parameter valuation in $S(\varphi_n)$ is doubly exponential in $n$.
\end{proposition}
  \vspace{0.1cm}

\noindent \textbf{Solving model checking.} A \emph{fair} Kripke structure $\Ku_f$ is a  Kripke structure equipped with a
 set $S_f$ of accepting states. An infinite path of $\Ku_f$ is \emph{fair}  if it visits infinitely many times states in $S_f$.
 Assume that $\Ku_f$ is over the set of atomic propositions given by $\Prop\cup\{c\}$, and let $\Lab_c$ be the associated propositional labeling. A \emph{$c$-pumpable fair path} of $\Ku_f$ is a fair infinite path $\pi$ of $\Ku_f$ such that each infix of $\pi$ associated to a $c$-block of the trace $\Lab_c(\pi)$ visits some state at least twice. Let $\Ku=\tpl{\Prop,S, E,\Lab,s_0}$ be a Kripke structure over $\Prop$, $\varphi$  a $\PromptHS$ formula, and $\Au_{\neg c}=\tpl{2^{\Prop\cup\{c\}},Q,q_0,\delta,F}$ be the  B\"{u}chi $\NFA$
 of Proposition~\ref{prop:FromHStoBuchiNFA} accepting the models of the $\HS$  formula $\neg \Rel_c(\varphi)\wedge \Alt_c$ (note that we consider the negation of $\Rel_c(\varphi)$). We define the fair Kripke structure
 \[
 \Ku\times \Au_{\neg c} = \tpl{\Prop\cup\{c\},
 S\times Q\times 2^{\{c\}},(s_0,q_0,\emptyset),E_c,\Lab_c,S\times F\times 2^{\{c\}}}
 \]
where (i) $((s,q,C),(s',q',C'))\in E_c$ iff $(s,s')\in E$ and $q'\in\delta(q,C\cup \Lab(s))$, and  (ii) $\Lab_c(s,q, C)=\Lab(s)\cup C$.
By construction, the traces  associated to the fair infinite paths of $ \Ku\times \Au_{\neg c}$ correspond to the $c$-colorings $w'$ of the traces of $\Ku$ which are accepted by $\Au_{\neg c}$ such that $c\notin w'(0)$. The following lemma is similar to Lemma~4.2 in~\cite{KupfermanPV09} and provides a characterization of emptiness of the set $V(\Ku,\varphi)$ of parameter valuations.\vspace{0.1cm}

\begin{lemma}\label{lemma:FromModelCheckingPumpable}
$\Ku$ does not satisfy $\varphi$ $($i.e., $V(\Ku,\varphi)=\emptyset$$)$ \emph{iff} $\Ku\times \Au_c$ has a $c$-pumpable fair path.
\end{lemma}
\begin{proof}
For the right implication, assume that $V(\Ku,\varphi)=\emptyset$. We need to show that  $\Ku\times \Au_c$ has a $c$-pumpable fair path. Let $k=  |Q||S| +1$ and $\alpha$ be the parameter valuation defined by $\alpha(u)= 2k$ for each $u\in U$. Since
$V(\Ku,\varphi)=\emptyset$, there is a trace $w$ of $\Ku$ such that $(w,\alpha)\not\models \varphi$. Let $w'$ be the $k$-bounded $c$-coloring of $w$ such that each $c$-block of $w'$ has length exactly $k$ and $c\notin w'(0)$. Since $w'\models \Alt_c$ and $c(\varphi)=\Rel_c(\varphi)\wedge \Alt_c$, by Lemma~\ref{lemma:AlternatingColororing}(2), it follows that $w'\models \neg \Rel_c(\varphi) \wedge \Alt_c$. Hence, $w'$ is accepted by the B\"{u}chi automaton $\Au_{\neg c}$, and
by construction there is a  fair path $\pi$ of $\Ku\times \Au_{\neg c}$ whose trace is $w'$. Now, each infix of $\pi$ associated to a $c$-block of $w'$ has length $k=|Q||S|+1$. Moreover,
by construction, the third component $C$ of the states $(s,q,C)\in S\times Q \times 2^{{c}}$ along such an infix does not change. It follows that such an infix visits one state at least twice. Thus, $\pi$ is a $c$-pumpable fair path of $\Ku\times \Au_{\neg c}$.\vspace{0.1cm}

For the left implication, assume that $\Ku\times \Au_c$ has a $c$-pumpable fair path  $\rho$. Let
$\alpha$ be an arbitrary parameter valuation and $k= \max_{u\in U}\alpha(u)$. We need to show that there is a trace $w$ of $\Ku$ such that $(w,\alpha)\not\models \varphi$. Since $\rho$ is a  $c$-pumpable fair path, each infix of $\rho$ associated to a $c$-block of the trace $\Lab_c(\rho)$ visits some state at least twice. The corresponding cycle in the infix can be pumped $k$-times. It follows that there is a $c$-pumpable fair path $\rho'$ of $\Ku\times \Au_{\neg c}$ such that  the $c$-blocks of the associated trace $w'$  have length at least $k$. By construction, $w'$ is the $c$-coloring of some trace $w$ of $\Ku$ and $w'$ is accepted by $\Au_{\neg c}$, i.e.~$w'\models \neg \Rel_c(\varphi)\wedge \Alt_c$. Hence, $w'$ is a $k$-spaced coloring of $w$ and $w'\not\models c(\varphi)$. By Lemma~\ref{lemma:AlternatingColororing}(1), it follows that $(w,\alpha)\not\models\varphi$, and the result follows.
\end{proof}

 By Lemma~\ref{lemma:FromModelCheckingPumpable}, we deduce that if $V(\Ku,\varphi)\neq \emptyset$,
then for the parameter valuation $\alpha$ such that $\alpha(u)= 2(|Q||S|+1)$ for each $u\in \PU$, it holds that   $\alpha\in V(\Ku,\varphi)$.
Indeed if $\alpha\notin V(\Ku,\varphi)$,
 by the first part of the proof of Lemma~\ref{lemma:FromModelCheckingPumpable},  there is a
 $c$-pumpable fair path of $ \Ku\times \Au_{\neg c}$, which leads to the contradiction $V(\Ku,\varphi)=\emptyset$.
It is known that checking the existence of a $c$-pumpable fair path in a fair Kripke structure is \NLOGSPACE-complete~\cite{KupfermanPV09}.
Recall that if $\varphi$ is a $\PromptF{\ABBbarBbarW}$ formula, then $c(\varphi)$ is
 a $\ABBbarBbarW$ formula, and by Proposition~\ref{prop:FromHStoBuchiNFA}, the size of $\Au_{\neg c}$ is doubly exponential in the size of $\varphi$. Thus, since both  $\Au_{\neg c}$ and $ \Ku\times \Au_{\neg c}$ can be built on the fly,
  by Lemma~\ref{lemma:ReductionToPrompness}, Proposition~\ref{prop:FromHStoBuchiNFA}, and Lemma~\ref{lemma:FromModelCheckingPumpable}, we obtain the following result.

 \begin{theorem}\label{theo:PHSModelChecking} Model checking against $\PHS$ is decidable. Moreover, model checking a   Kripke
  structure $\Ku$ against a  $\PHSF{\ABBbarBbarW}$ formula $\varphi$ is \EXPSPACE-complete and, in case $V(\Ku,\varphi)\neq\emptyset$,  there is a parameter valuation in  $V(\Ku,\varphi)$ which is bounded doubly exponentially in $|\varphi|$ and linearly in the number of $\Ku$-states.
 \end{theorem}

Similarly to the satisfiability problem for $\PHSF{\ABBbarBbarW}$, for each $n\geq 1$, we provide a lower bound of $2^{2^{n}}$ on the minimal parameter valuation for which a fixed Kripke structure satisfies a $\PHSF{\AB}$ formula by using a $\PHSF{\AB}$ formula
of size polynomial in $n$. For each $n\geq 1$, let $\psi_n$
be the $\AB$ formula over $\Prop= \{\#_1,\#_2,\$,0,1 \}$ in Lemma~\ref{lemma:lowerBoundParameters} whose unique model is the $n$-trace. One can trivially define a Kripke structure $\Ku$ over $\Prop$ whose set of traces consists of the traces whose first position  has label $\{\#_1,\#_2, 0\}$. Let us consider the  $\PHSF{\AB}$
formula $\varphi_n$ with just one parameter given by $\varphi_n\DefinedAs \psi_n \longrightarrow \hsA_{\leq u} \hsA(\len{1}\wedge\$)$.
Evidently,
by Lemma~\ref{lemma:lowerBoundParameters}, $V(\Ku,\varphi_n)$ is not empty and the minimal parameter valuation
in $V(\Ku,\varphi_n)$ is doubly exponential in $n$.  Hence, we obtain the following result.

\begin{proposition} There is a Kripke structure $\Ku$ over a  set  $\Prop$ of atomic propositions  and  a family $\{\varphi_n\}_{n\geq 1}$ of $\PHSF{\AB}$ formulas over $\Prop$ with just one parameter such that for each $n\geq 1$, $\varphi_n$ has size polynomial in $n$, $V(\Ku,\varphi_n)\neq \emptyset$,  and
the smallest parameter valuation in $V(\Ku,\varphi_n)$ is doubly exponential in $n$.
\end{proposition}

\section{Conclusion}

We have introduced parametric $\HS$ ($\PHS$), a parametric extension of the interval temporal logic 
$\HS$ under the trace-based semantics. The novel logic
allows to express parametric timing constraints on the duration  of the intervals. We have shown 
that the satisfiability and model checking problems for the whole logic are decidable, and for the fragment 
$\PHSF{\ABBbar}$ of $\PHS$, the problems are \EXPSPACE-complete. Moreover, for the fragment
$\PHSF{\ABBbar}$, we gave tight bounds on optimal parameter values for the considered problems. An intriguing open 
question is the expressiveness of $\PHSF{\ABBbar}$ (or more in general $\PHS$) versus parametric $\LTL$ ($\PLTL$).
We have shown that $\PHSF{\ABBbar}$ subsumes $\PLTL$. In particular, given a $\PLTL$ formula $\varphi$, it is possible to construct in linear time a $\PHSF{\ABBbar}$ on the same set of parameters which is equivalent to $\varphi$ \emph{for each parameter valuation}.
Is $\PHSF{\ABBbar}$ more expressive than $\PLTL$? Another problem left open is whether \PromptHS\ is strictly less expressive than full \PHS.

\bibliographystyle{eptcs}
\bibliography{bib2}
 

\end{document}